%% file: main.tex
\documentclass[letterpaper, 10 pt, journal, twosid]{ieeetran}

\IEEEoverridecommandlockouts    
\usepackage[hyphens]{url}
\usepackage{hyperref}
\hypersetup{colorlinks=false,breaklinks=true}
\usepackage{amsfonts,amsthm,amssymb,mathrsfs}
\usepackage{amsmath,mathtools}
\usepackage{algorithmic}
\usepackage{booktabs,tabularx}
\usepackage{cases}
\usepackage[ruled,vlined]{algorithm2e}
\usepackage{array}
\usepackage{xcolor}
\usepackage{textcomp}
\usepackage{stfloats}
\usepackage{url}
\usepackage{verbatim}
\usepackage{graphicx}
\usepackage[utf8]{inputenc}
\usepackage[english]{babel}
\usepackage{array}
\usepackage{adjustbox}
\usepackage{cite}
\usepackage{xcolor}
\usepackage{tabularx}
\usepackage{enumitem}
\usepackage{marginnote}
\usepackage{multirow}
\def\BibTeX{{\rm B\kern-.05em{\sc i\kern-.025em b}\kern-.08em
    T\kern-.1667em\lower.7ex\hbox{E}\kern-.125emX}}
\usepackage{balance}

\newtheorem{theorem}{Theorem}

\newtheorem{lemma}{Lemma}

\newtheorem{remark}{Remark}

\newtheorem{standing}{Standing Assumption}

\newcommand{\mc}{\mathcal}
\newcommand{\R}{\mathbb{R}}

\newcommand{\N}{\mathbb{N}}

\usepackage{stfloats}
\usepackage{soul}
\usepackage{subfigure}

\usepackage{booktabs}
\usepackage{footnote}
\usepackage{tabularx}
\usepackage{threeparttable} 
\usepackage{multirow} 
\usepackage{rotating} 
\usepackage{bigstrut} 
\usepackage{pifont}

\usepackage{adjustbox}

\newcolumntype{R}[2]{%
	>{\adjustbox{angle=#1,lap=\width-(#2)}\bgroup}%
	l%
	<{\egroup}%
}

\usepackage{todonotes}
\newboolean{showcomments}
\setboolean{showcomments}{true}         

\ifthenelse{\boolean{showcomments}}
{\newcommand{\nb}[2]{
		\fbox{\bfseries\sffamily\scriptsize#1}
		{\sf\small$\blacktriangleright$\textit{\textcolor{red}{#2}}$\blacktriangleleft$}
	}
}
{\newcommand{\nb}[2]{}
	
}

\newcommand{\lambdamin}{\lambda_{\mathrm{min}}}
\newcommand{\lambdamax}{\lambda_{\mathrm{max}}}

\usepackage{scalerel}
\usepackage{tikz}

\begin{document}

\title{Counter-example guided inductive synthesis of control Lyapunov functions for uncertain systems}
\author{Daniele Masti, Filippo Fabiani, Giorgio Gnecco, and Alberto Bemporad
\thanks{
The authors are with the IMT School for Advanced Studies Lucca, Piazza San Francesco 19, 55100, Lucca, Italy 
({\tt \footnotesize \{daniele.masti, filippo.fabiani, giorgio.gnecco, alberto.bemporad\}@imtlucca.it}).
The first author was supported in part by Consorzio Interuniversitario Nazionale per l'Informatica (CINI) through  Research Project under Grant CA 01/2021 a.i. 2. The first and third authors are members of INdAM-GNAMPA.
}
}


\maketitle         
\thispagestyle{empty}
\pagestyle{empty}    


\begin{abstract}
We propose a counter-example guided inductive synthesis (CEGIS) scheme for the design of control Lyapunov functions and associated state-feedback controllers for linear systems affected by parametric uncertainty with arbitrary shape. 
In the CEGIS framework, a learner iteratively proposes a candidate control Lyapunov function and a tailored controller by solving a linear matrix inequality (LMI) feasibility problem, while a verifier either falsifies the current candidate by producing a counter-example to be considered at the next iteration, or it certifies that the tentative control Lyapunov function actually enjoys such feature. We investigate the Lipschitz continuity of the objective function of the global optimization problem solved by the verifier, which is key to establish the convergence of our method in a finite number of iterations. Numerical simulations 
confirm the effectiveness of the proposed approach.  
\end{abstract}



\maketitle

\section{Introduction}

\input{intro.tex}
\section{Problem definition and preliminaries}
\label{sec:background}
In this paper we aim at designing control Lyapunov functions and related state-feedback controllers for discrete-time
uncertain linear systems of the form:
\begin{align}
\Sigma : 
x(k+1)=A(k) x(k) +B(k) u(k),
\label{eq:Sigma}
\end{align}
where $x \in \mathbb{R}^n$ and $u \in \mathbb{R}^m$ denote the state and control vectors, respectively,
and $x$ is assumed fully available.
The dynamical matrices are so that $(A(k),B(k)) \in \Omega \subseteq \mc A \times \mc B$, for all $k \ge 0$, where $\mc A \subseteq \R^{n\times n}$ and $\mc B \subseteq \R^{n\times m}$. 
Next, we introduce the main assumption adopted throughout:

\begin{standing}\label{standing:ass}

    $\Omega$ is a nonempty compact set.
    \hfill$\square$
\end{standing}

As made evident later in the paper, a distinct feature of our approach will be the possibility to efficiently handle uncertain systems as in \eqref{eq:Sigma} i) without necessarily imposing any { convex polytopic 
structure} to characterize the uncertainty,
and ii) sampling a (possibly large) number of matrices in $\mc A$ and $\mc B$.
Even though our CEGIS scheme will require the availability of $\Omega$ to solve the verifier task, we remark here that 
it makes 
no (possibly conservative) outer-approximation of such a set. 



We now recall some key notions on Lyapunov stability for discrete-time linear uncertain systems as $\Sigma$ in \eqref{eq:Sigma}. 
In particular, such a system is said to be \emph{stable} in some set $\mc X \subseteq \R^n$, with $0 \in \mc X$, if there exist a function $V : \mathbb{R}^n \to \mathbb{R}$ and some control law $u:\R^n\to\R^m$ so that \cite[Ch.~2.3]{blanchini2015set}:
\begin{equation}
\begin{aligned}
V(x)> 0, & ~\forall x \in \mc X \setminus \{0\},\\
V(0) = 0 & \iff x = 0, \\
V(x^+)-V(x)\leq 0, & \; \forall (x, x^+),
\end{aligned}
\label{eq:LyapunovCondition}
\end{equation}
where $(x, x^+)$ denotes a generic state-successor pair, i.e., $x^+ = A x + B  u(x)$, for any possible $(A,B) \in \Omega$. 
In this case, $V(\cdot)$ represents a (common) \emph{control Lyapunov function} with associated state-feedback controller $u(\cdot)$, and ensures closed-loop stability in $\mc X$
to the trajectories originating from $\Sigma$ in \eqref{eq:Sigma} with matrices
$(A(k),B(k))$ varying in $\Omega$. 
When the third condition holds as strict inequality, one may additionally recover convergence to the origin of the 
closed-loop trajectories.

To preserve mathematical tractability while designing functions and controllers so that conditions \eqref{eq:LyapunovCondition} are met, a common choice is to rely on a quadratic form for $V(\cdot)$, i.e.,
$
V(x)=x^\top Px,
$
with $P \in \mathbb{S}^{n}_{\succ 0}$, and a linear one for the controller, $u(x)=Kx$, for some gain matrix $K\in\R^{m\times n}$. Such a linear-quadratic parametrization for $V(\cdot)$ and $u(\cdot)$ is the one that will also be employed in the remainder. 

We finally recall some facts regarding the well-known convex polytopic case, namely when $\Omega =\textrm{conv}(\{(A_1,B_1), \ldots, (A_M,B_M)\})$ with $M$ given generator pairs of  matrices, that will be instrumental for the learner task.
In that case, a control Lyapunov function with associated robust control law for $\Sigma$ can be designed by solving an LMI feasibility problem \cite[Ch.~7.3.3.2]{duan2013LMIs}:
\begin{equation}\label{eq:baseLineK}
\left\{
\begin{aligned}
&\underset{X,W}{\mathrm{min}} && 0\\
& ~\mathrm{ s.t.} &&
\begin{bmatrix}
X& X A^\top_j +W^\top B_j^\top\\ \star & X
\end{bmatrix} \succ 0,~j = 1,\ldots,M.
\end{aligned}
\right.
\end{equation}

Given a solution pair $(X^\star, W^\star)$ to \eqref{eq:baseLineK} one hence obtains:
\begin{itemize}
\item A linear controller robustly stabilizing $\Sigma$ to the origin:
\begin{equation}
u(x)=Kx=W^\star(X^\star)^{-1}x;
\label{eq:staticFeedback}
\end{equation}
\item A control Lyapunov function $V(x)=x^\top Px$, $P=(X^\star)^{-1}$, for the system \eqref{eq:Sigma} with linear controller \eqref{eq:staticFeedback}.
\end{itemize}

\section{Control Lyapunov functions design through counter-examples}
\label{sec:cegis}

\subsection{The learner task}\label{subsec:learner}

The considerations above regarding \eqref{eq:baseLineK} directly apply to the CEGIS-based framework we are about to introduce. In fact, let $\mc C_i$ be the set of \textit{counter-examples} (or simply \textit{samples}) of dynamical matrices $(\hat A_j,\hat B_j)$ considered at the $i$-th iteration, $j=1,\ldots,i$. We rewrite the LMI constraints in \eqref{eq:baseLineK} to account for the elements of $\mc C_i$ as follows:
\begin{equation}\label{eq:LMICon2}
\left\{
\begin{aligned}
&\underset{X,W}{\mathrm{min}}&&0\\
&~\mathrm{ s.t. }&&
\begin{bmatrix}
X&X \hat A_h^\top+ W^\top \hat B_h^\top\\ \star & X
\end{bmatrix} \succcurlyeq \varepsilon I,~h=1,\ldots,M_i, \\
&&& X \preccurlyeq \eta I,~ W \in \mc W,
\end{aligned}
\right.
\end{equation}



\noindent where, in this case, the index $h$ allows us to enumerate the vertex matrices of $\mc C_i$, with $M_i \coloneqq |\textrm{vert}(\textrm{conv}(\mc C_i))|$, where clearly $M_i\le i$ and, for any $h=1,\ldots,M_i$, $(\hat A_h, \hat B_h) \in \textrm{vert}(\textrm{conv}(\mc C_i)) \subseteq \Omega$. The hyperparameters $\eta \ge \varepsilon>0$ can be chosen arbitrarily and are meant to upper (respectively, lower) bound the largest (smallest) eigenvalue of matrix $X$. Besides imposing some sort of contraction through strict positive-definiteness with $\varepsilon$, we will see that 
 $\eta$ and $\varepsilon$ are hyperparameters of our CEGIS method that strongly affect its performance. 
Finally, we assume $\mc W \subset \R^{m\times n}$ to be a nonempty compact set for technical reasons clarified later, whereas we also assume the convexity of $\mc W$ in order to keep the problem \eqref{eq:LMICon2} convex.

At each iteration $i$ of our CEGIS scheme, the learner accomplishes two tasks: first, given the current set $\mc C_i$ of samples of matrix pairs, it identifies those counter-examples that are vertices of its convex hull, for a total of $M_i$ vertices; leveraging this information, it then solves the resulting convex problem \eqref{eq:LMICon2} with LMI constraints. 

Note that the same discussion made in \S \ref{sec:background} also applies here. Specifically, solving \eqref{eq:LMICon2} produces a linear gain $K_i$ as in \eqref{eq:staticFeedback} that makes $V_i=x^\top (X^\star)^{-1} x$ a control Lyapunov function for the surrogate convex polytopic system identified by generator matrices $(\hat A_h,\hat B_h) \in \textrm{vert}(\textrm{conv}(\mc C_i))$.
On the other hand, since \eqref{eq:LMICon2} heavily depends on the choice of hyperparameters $\eta \ge \varepsilon>0$, in case the synthesis process fails, this implies that the parametrization adopted does not allow to design a control Lyapunov function with associated controller for $\Sigma$, with the assigned $(\eta, \varepsilon)$. We will discuss possible solutions to this issue later in \S \ref{subsec:cegis}.



\subsection{The verifier task}\label{subsec:verifier}
The verifier aims at certifying that a given candidate function and associated controller, designed by exploiting samples in $\mc C_i$, define a control Lyapunov function 
for $\Sigma$ with $(A(k),B(k)) \in \Omega$.
Thus, by specializing the third condition in \eqref{eq:LyapunovCondition} to a quadratic function and linear controller, we get: 
\begin{align}\label{eq:quadratic}
&V(x^+)- V(x)= (x^+)^\top P x^+ - x^\top Px \notag \\
=& ~ x^\top ((A^\textrm{cl})^\top P A^\textrm{cl}-P)x \leq 0,~\forall x \in \mc X,~ (A,B) \in \Omega,
\end{align} 
where $A^\textrm{cl}\coloneqq A+BK$. In particular, for a given $P$ and gain matrix $K$, condition~\eqref{eq:quadratic} is verified for all $x \in \mc X$ whenever $(A^\textrm{cl})^\top P A^\textrm{cl}-P$ is negative semi-definite for all $(A,B) \in \Omega$. 
Thus, fixing $(A,B) \in \Omega$ and applying the Schur complement to \eqref{eq:quadratic} yields the 
symmetric matrix
$\Xi(A^\textrm{cl})\coloneqq\left[\begin{smallmatrix}
P & (A^\textrm{cl})^\top P \\
\star &P
\end{smallmatrix}\right]$,
 whose eigenvalues are therefore all real. 
Hence, 
checking whether a given matrix $P=P^\top\succ0$ produces a control Lyapunov function for $\Sigma$ is 
equivalent to checking non-negativity of the optimal value $\lambda^\star$ of the following optimization problem:
\begin{align}\label{eq:smallestEig}
\lambda^\star=\underset{(A,B) \in \Omega}{ \mathrm{min}} & \;  \lambda_{\mathrm{min}}(\Xi(A^\textrm{cl})),
\end{align}
where $\lambdamin(\Xi(\cdot))$ denotes the smallest eigenvalue of $\Xi(\cdot)$.

Therefore, at the $i$-th iteration of the procedure, the verifier takes the matrices $P_i = (X^\star)^{-1}$ and $K_i=W^\star(X^\star)^{-1}$ proposed by the learner solving \eqref{eq:LMICon2}, defines $A^\textrm{cl}_i \coloneqq A+BK_i$, and
$
\Xi(A^\textrm{cl}_i) =
\left[
\begin{smallmatrix}
    P_i & (A^\textrm{cl}_i)^\top P_i\\
    \star & P_i
    \label{eq:closedLoopCondition}
\end{smallmatrix}\right]
$
 as in the definition of $\Xi(A^\textrm{cl})$, where the subscript $i$ makes explicit the dependence on the learner tentative solution, and finally  solves \eqref{eq:smallestEig}. If $\lambda^\star <0$, then a minimizer pair $(A^\star, B^\star) = (\hat A_{i+1}, \hat B_{i+1})$ is taken as a counter-example. Otherwise, no counter-example exists. 


Note that, however, \eqref{eq:smallestEig} amounts to solving a generally non-convex optimization problem 
that requires the availability of the set $\Omega$. Moreover, a global minimum must be attained to prove that no counter-example exists. On the other hand, it also features several properties that make finding a globally optimal solution an affordable task. To see this, we need the following result borrowed from 
matrix perturbation theory:

\begin{lemma}[\hspace{-.01cm}\cite{hornjohnson1991}]
Consider two symmetric matrices $K$ and $L$ of the same dimension. Then, it holds that
$
| \lambdamin(K)-\lambdamin(K+L)| \leq || L||_{\mathrm{op}}$,
where $|| \cdot||_{\mathrm{op}}$ is the operator norm of a matrix,  induced by the $l_2$-norm.
\hfill$\square$
\label{lemma:liplamdamin}
\end{lemma}
In other words, the minimum eigenvalue $\lambdamin(\cdot)$ of a symmetric matrix is a Lipschitz continuous map of the elements of the matrix. This is a consequence of Weyl's inequalities \cite{hornjohnson1991}. 

It is recalled here that any linear operator between two finite-dimensional normed spaces is bounded (this applies, in particular, to the operator $\Xi(\cdot)$, which maps some $A^\textrm{cl}$ to $\Xi(A^\textrm{cl})$). Since bounded linear operators preserve Lipschitz continuity, we are then able to prove the following result: 
\begin{theorem}
Let $P = (X^\star)^{-1}$ and $K=W^\star(X^\star)^{-1}$ with $(X^\star,W^\star)$ be a solution to \eqref{eq:LMICon2} with a given $\eta\ge\varepsilon>0$. Then, there exists some constant $\ell = \ell(\varepsilon) \ge 0$ such that  
$$
| \lambdamin(\Xi(A^\mathrm{cl}))-\lambdamin(\Xi(A^\mathrm{cl}+\Delta A^\mathrm{cl}))| \leq \ell || \Delta A^\mathrm{cl}||_{\mathrm{op}},
$$
with $A^\mathrm{cl}=A+BK$ for $(A,B)\in\Omega$ and perturbation $\Delta A^\mathrm{cl}$ induced by some pair $(\Delta A, \Delta B)$.
\hfill$\square$
\label{theo:LipSigma}
\end{theorem}
\begin{proof}
 One gets:
$$|\lambdamin(\Xi(A^\textrm{cl}))-\lambdamin(\Xi(A^\textrm{cl}+ \Delta A^\textrm{cl}))| \leq \left \|\begin{bmatrix}
0 & (\Delta A^\textrm{cl})^\top P \\
\star &0
\end{bmatrix} \right\|_{\rm op}$$
where 
$\|\cdot \|_{\rm op}$-norm is the matrix norm induced by the $l_2$ one. This leads to:
$$
\begin{aligned}
&\left\|\begin{bmatrix}
0 & (\Delta A^\textrm{cl})^\top P \\
\star &0
\end{bmatrix} \right\|_{\rm op} =\underset{\|x\|_2=1}{\textrm{sup}} \left\|\begin{bmatrix}
0 & (\Delta A^\textrm{cl})^\top P& \\
\star &0
\end{bmatrix} \begin{bmatrix}
x_1 \\
x_2
\end{bmatrix}\right\|_2 & \\
&=\underset{\|x\|_2=1}{\textrm{sup}}\left\|\begin{bmatrix}
(\Delta A^\textrm{cl})^\top P x_2 \\
P^\top \Delta A^\textrm{cl} x_1
\end{bmatrix}\right\|_2 \\
&= \underset{\|x\|_2=1}{\textrm{sup}} \sqrt{\|(\Delta A^\textrm{cl})^\top P\|_{\rm op}^2 \|x_2\|_2^2 + \|P^\top \Delta A^\textrm{cl}\|_{\rm op}^2 \|x_1\|_2^2} & \\
&=\|(\Delta A^\textrm{cl})^\top P\|_{\rm op} \leq \|\Delta A^\textrm{cl}\|_{\rm op} \|P\|_{\rm op} \leq \varepsilon^{-1}\|\Delta A^\textrm{cl}\|_{\rm op},
\end{aligned}
$$
which follows from the fact that 
$P=P^\top\succ0$ and  
$P=(X^\star)^{-1}$, for 
$X^\star$ solving~\eqref{eq:LMICon2}. Thus, setting $\ell = \ell(\varepsilon) \coloneqq \varepsilon^{-1}$ concludes the proof.
\end{proof}

\vspace{-0.2cm}

{
\begin{remark}
 The bound in Theorem \ref{theo:LipSigma} can be expressed also in terms of a Lipschitz constant w.r.t. the pair $(A,B)$, since $\|\Delta A^{\rm cl}\|_{\rm op} \leq \max\{\|I\|_{\rm op},\|K\|_{\rm op}\} \|\Delta A \,\,\Delta B\|_{\rm op}$ as $\Delta A^{\rm cl}=\Delta A + \Delta B K$. Thus, such a Lipschitz constant is bounded from above by $\varepsilon^{-1} \max\{1,\|K\|_{\rm op}\}$. A uniform upper bound, i.e., independent from $K$, could be also obtained by accounting for $K=W^\star (X^\star)^{-1}$, $(X^\star)^{-1}=P$, and $W^\star \in \mc W$, which is compact. In case $B$ is not uncertain (i.e., $\Delta B=0$), the upper bound on the Lipschitz constant w.r.t. $A$ reduces to $\varepsilon^{-1}$. 
 \hfill$\square$
\end{remark}
}
We 
note that Lemma~\ref{lemma:liplamdamin} and Theorem~\ref{theo:LipSigma} 
suggest that global optimization techniques, tailored for Lipschitz optimization, can be used to verify if a candidate function $V_i$ is a Lyapunov function for $\Sigma$. We detail this crucial point
in the next section.

\vspace{-0.1cm}

\subsection{The proposed CEGIS scheme}\label{subsec:cegis}

Algorithm~\ref{algo:OverallCegis} summarizes the main steps of the proposed CEGIS-based iterative scheme for the design of a control Lyapunov function, accompanied with a suitable linear state-feedback controller, for the linear uncertain system $\Sigma$ in \eqref{eq:Sigma}. In it, black-filled bullets refer to tasks performed by the learner, while the empty bullet to the one performed by the verifier. 


\vspace{-0.2cm}

 \begin{algorithm}[h!t]
	\caption{CEGIS-Lyapunov learning method}\label{algo:OverallCegis}
	\DontPrintSemicolon
	\SetArgSty{}
	\SetKwFor{ForAll}{for all}{do}{end forall}
	\textbf{Initialization:} Set $\eta \ge \varepsilon>0$, fill $\mc C_1 = \{(\hat A_1, \hat B_1)\}$ with some $(\hat A_1, \hat B_1)\in\Omega$. \\
	\textbf{Iteration $(i \in \N^+)$:} 
        \begin{itemize}\setlength{\itemindent}{-.2cm}
            \item[$\bullet$] Identify $\textrm{vert}(\textrm{conv}(\mc C_i))$
        \end{itemize}
        \begin{itemize}\setlength{\itemindent}{-.2cm}
            \item[$\bullet$] Solve \eqref{eq:LMICon2}, set $P_i=(X^\star)^{-1}$, $K_i=W^\star(X^\star)^{-1}$
            
            \textbf{If} \eqref{eq:LMICon2} \texttt{infeasible} \textbf{then} \texttt{exit}
        \end{itemize}
        \begin{itemize}\setlength{\itemindent}{-.2cm}
            \item[$\circ$] Solve \eqref{eq:smallestEig} using a Lipschitz global optimization tool
            
            \textbf{If} $\lambda^\star<0$ \textbf{then} $\mc C_{i+1}\leftarrow \mc C_i \cup \{(A^\star,B^\star)\}$, \texttt{repeat}\\
            \textbf{If} $\lambda^\star\ge0$ \textbf{then} $V=x^\top P_i x$, $K=K_i$, \texttt{exit}
        \end{itemize}
\end{algorithm}

\vspace{-0.2cm}

Note that Algorithm~\ref{algo:OverallCegis} may terminate in two different points with two possible outcomes. According to the discussion in \S \ref{subsec:learner}, the learner has to solve an LMI feasibility problem, which is not only susceptible to the parametrization adopted and the quality of the samples the verifier produces, but also on the choice of the hyperparameters $\eta \ge \varepsilon>0$. For these reasons, \eqref{eq:LMICon2} may 
not be feasible, and this coincides with the notion of infeasibility mentioned in the next Theorem~\ref{th:convergence}, namely that a quadratic control Lyapunov function and associated linear controller do not exist for the uncertain system $\Sigma$ in \eqref{eq:Sigma} { with} the assigned $\eta$ and $\varepsilon$. 
{
 Note that, however, this does not necessarily imply that such a pair does not exist for the underlying system, for other choices of the hyperparameters. In fact,} one might want to reinitialize the procedure with, e.g., a smaller 
$\varepsilon$, a larger 
$\eta$, or a larger $\mc W$. { Moreover, in view of the finite-time convergence guarantees featured by Algorithm~\ref{algo:OverallCegis} (detailed in the following Theorem~\ref{th:convergence}), the choice of $(\hat A_1, \hat B_1) \in \Omega$ may affect the number of steps needed, and also the specific control Lyapunov function generated, though not the actual outcome of Algorithm~\ref{algo:OverallCegis} (infeasibility or convergence are, indeed, independent from $(\hat A_1, \hat B_1)$).
}

Finally, in case the verifier
simply terminates without founding any counter-example, Algorithm~\ref{algo:OverallCegis} has converged. Such a convergence property is itself a certification of the correctness of the control Lyapunov function and associated linear controller found. Otherwise, if some counter-example has been found, the set $\mc C_{i+1}$ includes the new pair $(A^\star,B^\star)=(\hat A_{i+1}, \hat B_{i+1})$ obtained from \eqref{eq:smallestEig} and the iterative procedure continues, although not indefinitely, as shown next: 
\begin{theorem}\label{th:convergence}
    Let $\eta \ge \varepsilon>0${, and some $(\hat A_1, \hat B_1)\in\Omega$} be given. Then, in a finite number of iterations Algorithm~\ref{algo:OverallCegis} either declares infeasibility
    or produces a pair $(\bar P, \bar K)$ so that $V=x^\top \bar P x$ is a control Lyapunov function for the linear uncertain system $\Sigma$ in \eqref{eq:Sigma} with linear controller $u=\bar Kx$.
    \hfill$\square$
\end{theorem}
\begin{proof}
    We split the proof in two parts. First, we show that the verifier can only generate counter-examples through \eqref{eq:smallestEig} 
    outside the set of matrices considered by the learner to 
    find a solution to \eqref{eq:LMICon2}. This fact 
    is then exploited 
    to prove that the verifier can not produce counter-examples indefinitely, 
    therefore Algorithm~\ref{algo:OverallCegis} either returns infeasibility, or a solution pair $(\bar P, \bar K)$ in a finite number of steps.
    
    i) 
    Assume that \eqref{eq:LMICon2} is feasible at the generic $i$-th iteration, yielding a solution pair $(X^\star,W^\star)$. Then, by pre and post-multiplying by $\mathrm{diag}((X^\star)^{-1}, (X^\star)^{-1})$ both sides of the first 
    set of LMI constraints in \eqref{eq:LMICon2}, which is always possible in view of $X^\star \succcurlyeq \varepsilon I$ with $\varepsilon>0$, we obtain:
    $$
    \begin{aligned}
    &\begin{bmatrix}
        (X^\star)^{-1}&(\hat A_h+ \hat B_h W^\star (X^\star)^{-1})^\top(X^\star)^{-1} \\
        \star & (X^\star)^{-1}
    \end{bmatrix} \\
    &\hspace{1.8cm}
    \succcurlyeq  \varepsilon
    \mathrm{diag}((X^\star)^{-1}, (X^\star)^{-1})^{2} 
    \succcurlyeq  \varepsilon \lambdamin^2((X^\star)^{-1}) I \\
    &\hspace{1.8cm}
    \succcurlyeq  (\varepsilon/\lambdamax^2(X^\star)) I \succcurlyeq  (\varepsilon/\eta^2) I,
    \end{aligned}
    $$
    for all $h=1,\ldots,M_i$.
    This follows from 
    \eqref{eq:LMICon2}, since $\lambdamax(X^\star) \le \eta$ in view of $X^\star \preccurlyeq \eta I$,
    and $\lambdamin(X^\star)\ge\varepsilon$, 
    yielding $\varepsilon/\eta^2 > 0$ as $\eta\ge\varepsilon>0$. 
    In view of 
    $P_i=(X^\star)^{-1}$ and $K_i=W^\star (X^\star)^{-1}$, from the previous inequalities we have:
    \begin{equation}\label{eq:lower_bound}
    \Xi(A^{\textrm{cl}}_{i,h}) \succcurlyeq  (\varepsilon/\eta^2) I, \text{ for all } h=1,\ldots,M_i,
    \end{equation}
    with $A^{\textrm{cl}}_{i,h} = \hat A_h + \hat B_h K_i$, which implies $\lambdamin(\Xi(A^{\textrm{cl}}_{i,h})) \ge \varepsilon/\eta^2 > 0$ for all $h=1,\ldots,M_i$. However, since each pair $(\hat A_h, \hat B_h) \in \textrm{vert}(\textrm{conv}(\mc C_i))$, standard arguments in robust control of polytopic systems \cite[Chapter~7.3.3.2]{duan2013LMIs} imply that this 
    relation
    not only holds 
    for all vertex matrices, but also for any matrix in their convex hull, 
    yielding $\lambdamin(\Xi(A^{\textrm{cl}}_i)) \ge \varepsilon/\eta^2 > 0$ being $A^{\textrm{cl}}_i$ now constructed by exploiting any pair of matrices $(A,B) \in \textrm{conv}(\{(\hat A_h,\hat B_h)\}_{h=1}^{M_i})=\textrm{conv}(\mc C_i)$. 
    Hence, the verifier can only find counter-examples through \eqref{eq:smallestEig} strictly outside $\textrm{conv}(\{(\hat A_h,\hat B_h)\}_{h=1}^{M_i})$, i.e., $(\hat A_{i+1}, \hat B_{i+1}) \in \Omega \setminus \textrm{conv}(\mc C_i)$. 
    
    ii) Armed with the previous result, once the learner has obtained from the verifier a new counter-example $(\hat A_{i+1}, \hat B_{i+1}) \in \Omega \setminus \textrm{conv}(\mc C_i)$, either it computes a new tentative solution pair $(P_{i+1},K_{i+1})$ considering a larger pool of samples, or Algorithm~\ref{algo:OverallCegis} terminates at the second step of the $(i+1)$-th iteration. If the latter happens, then Algorithm~\ref{algo:OverallCegis} returns infeasibility in a finite number of iterations, as claimed.
    On the other hand, what could happen instead is that the verifier may produce counter-examples indefinitely. 
    We will show in the following that this can 
 not be the case.

    In particular, we start by assuming that \eqref{eq:LMICon2} is feasible for all the considered iterations $i\in\N^+$, otherwise Algorithm~\ref{algo:OverallCegis} declares infeasibility in a finite number of steps. Then, by relying on Theorem~\ref{theo:LipSigma}, at the generic $i$-th iteration we can quantify how far from $\textrm{conv}(\mc C_i)$ the new counter-example $(\hat A_{i+1}, \hat B_{i+1})$ is. 
    Precisely, it falls outside 
    the next convex set of matrices generating closed-loop eigenvalues 
    guaranteed to be at most $(\varepsilon/\eta^2)$-distant from those belonging to $\textrm{conv}(\mc C_i)$:
    $$
    \begin{aligned}
        \mc S_{\bar A^{\textrm{cl}}_i} \coloneqq &\{A^{\textrm{cl}} \in \R^{n\times n}\mid A^{\textrm{cl}} = A+BK_i, A\in\R^{n\times n}, \\ &B\in\R^{n\times m}, ||\bar A^{\textrm{cl}}_i- A^{\textrm{cl}}||_\mathrm{op} \le \varepsilon^2/\eta^2, \\
        &\text{ for all } \bar A^{\textrm{cl}}_i = A_j+B_jK_i,~(A_j, B_j) \in \textrm{conv}(\mc C_i) \}.
    \end{aligned}
    $$

    For the sake of contradiction, assume that the verifier generates an infinite sequence $\{(\hat A_i, \hat B_i)\}_{i\in\N^+}$, with $(\hat A_i, \hat B_i) \in \Omega$. From the first part of the proof, we know that $\mc C_{i+1} \supset \mc C_i$, and therefore the infinite, monotonically increasing sequence of sets $\{\mc C_i\}_{i\in\N^+}$ admits the 
    limit set $\bar{\mc C} = \cup_{i\in\N^+} \mc C_i$, which enables us to extract from $\{(\hat A_i, \hat B_i)\}_{i\in\N^+}$ a subsequence of counter-examples converging to some $(\bar A, \bar B) \in \textrm{cl}(\textrm{conv}(\bar{\mc C})) \cap \Omega$. 
    Since the learner iteratively generates tentative pairs $\{(P_i, K_i)\}_{i\in\N^+}$, produced by  
     solutions $\{(X^\star, W^\star)\}_{i\in\N^+}$ to \eqref{eq:LMICon2} living in nonempty compact sets, we can additionally extract a convergent subsequence from $\{(P_i, K_i)\}_{i\in\N^+}$ with limit pair $(\bar P, \bar K)$. Let us now consider a limiting ``safe'' set $\mc S_{\bar A^{\textrm{cl}}}$, defined as:
    $$
    \begin{aligned}
        &\mc S_{\bar A^{\textrm{cl}}} \coloneqq \{A^{\textrm{cl}} \in \R^{n\times n}\mid A^{\textrm{cl}} = A+B\bar K, A\in\R^{n\times n}, \\ &\hspace{1.3cm}B\in\R^{n\times m}, ||\bar A^{\textrm{cl}}- A^{\textrm{cl}}||_\mathrm{op} \le \varepsilon^2/\eta^2, \\
        &\hspace{1.3cm}\text{ for all } \bar A^{\textrm{cl}} = A_j+B_j \bar K,~(A_j, B_j) \in \textrm{conv}(\bar{\mc C}) \}.
    \end{aligned}
    $$
    Leveraging this limiting set, we can conclude that not all the points of the convergent subsequence extracted from $\{(\hat A_i, \hat B_i)\}_{i\in\N^+}$ are counter-examples. In fact, by definition of limit, there exists some large enough index $i_c$ associated to the underlying subsequence so that, by starting from $(\hat A_{i_{c+1}}, \hat B_{i_{c+1}})$, all the counter-examples generate closed-loop eigenvalues, obtained with the limit pair $(\bar P, \bar K)$, 
    strictly closer than $(\varepsilon/\eta^2)$ from those generated by $(\hat A_{i_{c}}, \hat B_{i_{c}})$ with the same limit pair $(\bar P, \bar K)$. This 
    contradicts the first part of the proof, i.e., the assumption that they are counter-examples. 
    
    Hence, under the feasibility of \eqref{eq:LMICon2}, what can only happen is that the sequence $\{(\hat A_i, \hat B_i)\}$ converges to a pair $(\bar A, \bar B) \in \Omega$ in a finite number of iterations, say $N\in\N^+$, entailing that also $\{\mc C_i\}$ converges to a set $\mc C_N=\bar{\mc C}$ in $N$ steps, with $\bar{\mc C}$ so that $\Omega \setminus \mc S_{\bar A^{\textrm{cl}}} = \emptyset$. However, after $N$ iterations, the verifier can not 
    produce any counter-example by evaluating the tentative pair $(P_N, K_N)=(\bar P, \bar K)$ proposed by the learner. 
    
    Concluding, one of the two outcomes of Algorithm~\ref{algo:OverallCegis} is returned in a finite number of iterations, ending the proof.
\end{proof}

\begin{remark}
Compared to 
ray-shooting~\cite{alessio2007squaring}, which builds a conservative outer-approximation of $\Omega$ through geometric considerations, our approach covers the whole set $\Omega$ of possible counter-examples by exploiting samples contained in it, along with the ``safety'' guarantees provided by the sets $\mc S_{\bar A^{\textrm{cl}}_i}$. 
\hfill$\square$
\end{remark}

\section{Computational aspects}\label{sec:comp_aspect}


\subsection{Tradeoff induced by $\eta$ and $\varepsilon$}
\label{sec:rolevareps}
Besides strongly affecting the learner convex optimization problem \eqref{eq:LMICon2}, we note that  $\eta$ and $\varepsilon$, with $\eta\ge\varepsilon>0$, assume a key role also in characterizing the performance of both Algorithm~\ref{algo:OverallCegis} and the verification task \eqref{eq:smallestEig}. 

In fact, from Theorem~\ref{th:convergence}, specifically the lower bound \eqref{eq:lower_bound}, it is evident that setting $\eta$ as close as possible to $\varepsilon$, with this latter large enough, the convergence of Algorithm~\ref{algo:OverallCegis} may require a smaller number of steps, as this combination allows to ``erode'' (by means of the safe set $\mc S_{\bar A^{\textrm{cl}}_i}$) a larger portion of the set $\Omega\setminus \textrm{conv}(\mc C_i)$ at each iteration. In addition, this choice is also favourable to improve the performance of the technique employed to solve \eqref{eq:smallestEig}. Indeed, convergence of algorithms for Lipschitz-continuous global optimization problems, as the one in \eqref{eq:smallestEig}, is well-studied in the literature. As a general rule of thumb inferred, the smaller the Lipschitz constant, the better the guarantee on the performance of the underlying algorithm~\cite{malherbe2017global}. In view of Theorem~\ref{theo:LipSigma}, this hence suggests to choose larger values for the hypeparameter $\varepsilon$.

However, following the considerations above may undesirably complicate the learner's optimization problem, as the feasible set of \eqref{eq:LMICon2} would shrink, thus making the design of quadratic control Lyapunov functions and associated linear controllers a challenging task. This naturally induces a tradeoff in the choice of the pair $(\eta,\varepsilon)$, which shall hence be set according to the problem data at hand.

\subsection{Speeding up the verification process}
To partially overcome the computational requirements of the verification process, one might want to rely on sensitivity-based methods.
 Indeed, while in general the partial derivative 
$
\partial \lambda_i(\Xi(A^\textrm{cl}))/\partial A^\textrm{cl}_{(h,k)}
$ of the $i$-th eigenvalue $\lambda_i(\Xi(A^\textrm{cl}))$ with respect to the element  in position $(h,k)$ of $A^\textrm{cl}$ may be not defined when $\lambda_i$ changes multiplicity, as noted in~\cite{kato2013perturbation}, such an happenstance is quite uncommon. This fact suggests that exploiting such an information and trying to solve \eqref{eq:smallestEig} using a sensitivity-based optimization algorithm like an SQP or interior-point solver might be a valid heuristic to speed-up the overall verification process, by possibly accelerating the process of finding a new counter-example. In this case, however, we loose any kind of convergence guarantee, since finding a local minimum in \eqref{eq:smallestEig} may still lead to a counter-example, but if no counter-example can be found by a local solver, then nothing can be concluded. 
Nevertheless, a sensitivity-based approach may lead to a substantial speed improvement in the evaluation of \eqref{eq:smallestEig}. This motivates the next modification of the verifier's task in Algorithm~\ref{algo:OverallCegis}. Given some number of trials $N_t\in\N^+$, { and denoting by $\hat{\lambda}^{\star}$ an estimate of $\lambda^{\star}$ (which is also an overestimate, being \eqref{eq:smallestEig} a minimization problem)}:


\smallskip
\noindent \textbf{Set} $\texttt{counter}=0$\\
\noindent\textbf{While} $\texttt{counter}<N_t$
\begin{itemize}
\item[$\circ$] Solve \eqref{eq:smallestEig} by a sensitivity-based method to get $\lambda^{\star,  \textrm{sens}}$\\
    \smallskip
    \textbf{If} $\lambda^{\star,  \textrm{sens}} <0$ \textbf{then} $\hat{\lambda}^\star \leftarrow \lambda^{\star,  \textrm{sens}}$, \texttt{exit while}\\
    \textbf{Else} $\texttt{counter} \leftarrow \texttt{counter}+1$, set different initial conditions for sensitivity-based method
\end{itemize}
\noindent\textbf{End}\\
\noindent\textbf{If} $\texttt{counter}==N_t$ \textbf{then} Solve \eqref{eq:smallestEig} as in Algorithm~\ref{algo:OverallCegis}
\smallskip

The same intuition also holds for the \textit{refinement} of the { approximate} minimizer found by the global optimization procedure, as it is well-known that the \textit{quality} of the solution provided by most global optimization tools is often subpar. 

{
\begin{remark} 
We remark that numerical issues 
may produce only a rough estimate $\hat{\lambda}^{\star}$ of $\lambda^{\star}$. When $\hat{\lambda}^{\star}<0$, this would be not harmful, because the verifier would still be able to generate a counter-example. However, for $\hat{\lambda}^{\star} \geq 0$, a non-valid Lyapunov function could be generated (as in \cite[Example~8]{AhmedPeruffoAbate2020}), in case $\lambda^{\star}<0$. To avoid this, one may replace the condition $\lambda^{\star} \geq 0$ in  Algorithm~\ref{algo:OverallCegis} with, e.g., $\lambda^{\star} \geq \frac{\varepsilon}{2}$, and $\frac{\varepsilon}{2}$ larger than the desired precision. In this case, the main change in the proof of Theorem \ref{th:convergence} would be the generation of smaller safe sets.
\hfill$\square$
\end{remark}
}





\section{Numerical examples}
\label{sec:numerical}
\input{experimental.tex}

\section{Conclusion}
\label{sec:conclusions}
We have introduced a CEGIS-based method to design control Lyapunov functions, accompanied by suitable controllers, for linear systems with parametric uncertainties 
in an arbitrary nonempty compact set.
By building upon 
results in linear operator theory, our technique features appealing Lipschitz-continuity properties that enable the use of efficient solvers for global optimization and guarantee finite-time convergence. 
Future work will focus on { improving the rate of convergence through enlarging the ``safe'' sets, as well as  estimating the number of iterations needed, on the extension of our to other type of uncertainties, and on its further comparison with available techniques (e.g., those based on SOS methods \cite{WuPrajna2004})}.

\bibliographystyle{IEEEtran}
\bibliography{biblio}

\end{document}

%% file: intro.tex
The design of effective control laws guaranteeing closed-loop stability and performance is the quintessential problem in control theory.
Among the avenues to do so, Lyapunov methods~\cite{blanchini2015set} have traditionally attracted interest in the systems-and-control community in view of their flexibility and power. Indeed, several approaches have been developed to synthesize different forms of Lyapunov functions (e.g., quadratic/quartic, piecewise affine/quadratic) for various types of systems, a task that is often accomplished {through convex optimization~\cite{biswas2005survey, kothare1996robust,parrilo2000structured,WuPrajna2004})}, although other approaches have been explored based on, e.g., neural networks~\cite{gabyetal2022} or formal verification~\cite{munser2021synthesis}. 
A recent trend, instead, looks at the design of (control) Lyapunov functions using counter-example guided sampling-based approaches~\cite{ravanbakhsh2019learning,mcgough2010symbolic}, to possibly take advantage of the presence of baseline controllers already in place~\cite{khansari2017learning, mordatch2014combining}.
Among these, the \textit{counter-example guided inductive synthesis} (CEGIS) family of techniques has found widespread use and success {in 
recent years~\cite{solar2006combinatorial,abate2017automated,berger2022learning,chang2019neural,dai2020counter,Abate2023Neural,AhmedPeruffoAbate2020}}. 
The main idea behind CEGIS consists of designing a function $f: \mathbb{R}^n \to \mathbb{R}$ by exploiting the information contained in a growing dataset 
of \mbox{(counter-)}examples (hereinafter called $\mc C_i$ at step $i$), while making sure that $f$ 
belongs to some 
given
set of functions $\Theta$. 
To this end, CEGIS relies on an iterative adversarial procedure 
with two actors, a \textit{learner} and a \textit{verifier}.
In summary, at each iteration $i\in\mathbb{N}^+$: 
\begin{enumerate}
    \item The learner takes the dataset $\mc C_i$ as input and synthesizes a function $f_i$ from a hypothesis space $\mathcal{H}$, according to some criterion, or establishes that such a function can not be synthesized; 
    \item The selected function $f_i$ is then passed to the verifier, which verifies if  $f_i \in \Theta \coloneqq\{f: \mathbb{R}^n \to \mathbb{R} \mid g(f(z)) \le 0, \text{ for all } z \in \mc Z\}$, for some $g : \mathbb{R} \to \mathbb{R}$ 
    and set $\mc Z \subseteq \mathbb{R}^n$ over which the variable $z$ takes values. The verifier can hence take two different conclusions: 
    \begin{enumerate}
    \item It finds a \textit{counter-example} $z_{i+1} \in \mc Z$ so that $g(f_i(z_{i+1}))>0$. In such a case     
            $
            \mc C_{i+1} \leftarrow \mc C_i \cup \{z_{i+1}\},
            $
            and a new iteration step follows;

    \item It certifies that no counter-example exists, yielding the positive conclusion of the procedure.

    \end{enumerate}
\end{enumerate}


While CEGIS schemes have been successfully employed to design disparate types of (control) Lyapunov functions for different classes of systems \cite{solar2006combinatorial,abate2017automated,berger2022learning,chang2019neural,dai2020counter,Abate2023Neural},  little attention has been given to how such methods can be used to extend the capabilities of traditional robust control approaches, both in terms of enlarging their applicability and, more prominently, of reducing their computational~requirements.

In this paper we make an attempt to fill this gap. 
Specifically, we start from the classical machinery proposed in \cite{kothare1996robust}, 
and leverage a novel Lipschitz-based verification approach, rooted in perturbation of linear operators~\cite{kato2013perturbation}, to efficiently design control Lyapunov functions, and associated feedback controllers, for linear systems whose parameters are only known to belong to an arbitrary compact set. 
This is achieved without introducing 
{conservativeness due to the adoption of a convex polytopic outer-approximation of 
the 
uncertainty set. Indeed, due to computational reasons, standard methods use it to approximate the convex hull of such a set (which may have a complex shape), possibly based on a few vertices}. 
In summary, we make the following contributions:
\begin{itemize}
    \item We design a novel CEGIS-based method, in which the learner solves an LMI feasibility problem, while the verifier a Lipschitz continuous nonconvex one;
    \item We get Lipschitz continuity of the objective function of the verifier's task, which is key to resort on efficient solvers tailored for global Lipschitz optimization;
    \item We prove that our algorithm converges in finite-time.  
\end{itemize}

In \S \ref{sec:background} we formalize the control problem and recall fundamentals of Lyapunov stability, while in \S \ref{sec:cegis} we describe the learner' and verifier's tasks, characterize the technical properties of the problems they are asked to solve, and establish the finite-time convergence of our CEGIS-based method. Computational aspects are, instead, analyzed in \S \ref{sec:comp_aspect}, 
while numerical examples are finally given in \S \ref{sec:numerical}.

%% file: experimental.tex
We verified the effectiveness of the proposed approach on two numerical examples.
All experiments were carried out on a PC with an Intel core i7-1165G7, Ubuntu Linux and MATLAB R2022a. Solution of the optimization problem \eqref{eq:LMICon2} was obtained by 
YALMIP~\cite{Lofberg2004YALMIP}, while the vertices of polytopes were computed using 
MPT3~\cite{Herceg2013MPT3} and 
YALMIP~\cite{Lofberg2004YALMIP} 
As a Lipschitz global optimization solver, we adopted the \texttt{dDirect\_GLce} algorithm from the DIRECTGO toolbox~\cite{stripinis2022directgo}.
All the results of the CEGIS scheme were obtained in a few tens of seconds.

\subsection{Polytopic uncertainty set}
We first consider an uncertain system $\Sigma$ in \eqref{eq:Sigma} with four states and one input, where the matrix $B$ is not uncertain and $B(k)= B = \left[0 \; 0 \; 0 \; 1\right]^\top$, while { each entry of} the dynamical matrix $A$ is 
subject { -- independently from the other entries --} to an interval uncertainty, 
namely
$$
\begin{aligned}
    A(k) \geq &\left[\begin{smallmatrix}
   -0.6685&-0.8709&-0.2028&-1.5547\\
\phantom{-}1.1457&-0.5898&\phantom{-}0.5688&\phantom{-}0.8496\\
-0.7812&-0.5754&-0.8774&-0.2501\\
-1.1429&\phantom{-}0.1730&\phantom{-}0.7763&\phantom{-}0.1618\\
    \end{smallmatrix}\right]
,  \\
    A(k)\leq& \left[\begin{smallmatrix}
-0.6295&-0.8202&-0.1910&-1.4641\\
\phantom{-}1.2166&-0.5555&\phantom{-}0.6040&\phantom{-}0.9022\\
-0.7357&-0.5419&-0.8263&-0.2355\\
-1.0763&\phantom{-}0.1837&\phantom{-}0.8243&\phantom{-}0.1718\\
   \end{smallmatrix}\right],
\end{aligned}
$$
for all $k\in\N$. Although such { an uncertain system} could be reliably stabilized using 
polytopic techniques such as the one presented in~\cite{kothare1996robust}, this would require 
to solve an optimization problem with 
$2^{16}=65536$ LMI constraints, { each generated by taking independently} either the lower or upper bound for each of the 16 entries of the matrix $A$. While theoretically possible, this would be 
{ hardly doable even with} modern hardware. 

By using the proposed CEGIS method, instead, we 
found a quadratic control Lyapunov function with $P$ matrix:
$$
\bar P=\left[\begin{smallmatrix}
  162.4930 &  29.5126 &  82.0931 & 176.8625\\
   29.5126  & 42.9150 &  21.4642 &  50.7574\\
   82.0931  & 21.4642 &  58.6050 &  99.6742\\
  176.8625  & 50.7574 &  99.6742 & 232.9383\\
    \end{smallmatrix}\right]
$$
accompanied with linear feedback policy:
$
    \bar K=   \begin{bmatrix}
    1.7667 &   0.9014  & -0.3555    &1.0089
     \end{bmatrix}.
$
{ The soundness of the result was verified a-posteriori by checking the positive definiteness of $\Xi(A_h + B \bar K)$, built from the obtained $\bar P$, on each of the $65536$ pair of vertex matrices $(A_h,B)$ defining $\Omega$}. 
Such a result was obtained with a final constraint set $\mc C_i$ made of five samples (the initial point and four counter-examples). {
 On the other hand, on our reference hardware solving the original LMI feasibility problem \eqref{eq:baseLineK} was not possible due to limitations on the available computational power. 
 In contrast, the proposed method was able to synthesize a quadratic control Lyapunov function in a few tens of seconds. 
}
{

We note that the results shown above were obtained excluding from the set of Lyapunov candidates all those functions not satisfying the constraints dictated by our choice of $\eta$ and $\varepsilon$ in~\eqref{eq:LMICon2}. While, in our experience, tuning said quantities was not problematic, an unreasonable choice of those hyperparameters (e.g., $\eta=2 \cdot 10^3$ and $\varepsilon=10^3$) may indeed cause the proposed scheme to fail. In that case, one can always restart the procedure with more permissive values for $\eta$ and $\varepsilon$. 
Finally, when a quadratic control Lyapunov function is found, it may be characterized by a smaller $\varepsilon$ than in the verifier task. 
}


\vspace{-0.4cm}

\subsection{Spherical uncertainty set}
We consider now the case in which the matrix $A(k)$ is subject to spherical uncertainty, i.e.,
$
    \Omega = \{ A \in \R^{2\times 2} \mid (\mathrm{vec}(A)-c)^\top Q (\mathrm{vec}(A)-c)\ -1 \leq 0 \},
$
where the operator $\textrm{vec}(\cdot)$ stacks the columns of its argument into a single column vector, $Q=5 I$, while $B(k)= B = \left[0 \; 1\right]^\top$.
In particular, $c=\mathrm{vec}(A_\mathrm{centroid})$, where 
$
A_\mathrm{centroid}= \left[\begin{smallmatrix}
       \phantom{-} 0.6458   & 0.3852\\
   -1.4651   & 1.1183
\end{smallmatrix}\right]$. 
The synthesis of a control Lyapunov function was obtained with a counter-example set including four points, i.e., Algorithm~\ref{algo:OverallCegis} converged in three iterations, 
yielding 
$
\bar P = \left[\begin{smallmatrix}
  117.4770 &  60.7593\\
   60.7593 & 130.6819
\end{smallmatrix}\right]$ 
and 
$\bar K=\begin{bmatrix}
    0.9280 &  -1.4962    
\end{bmatrix}$. 
{ We verified the soundness of these results by using a convex polytopic outer-approximation of the spherical set $\Omega$, built using the method in~\cite{yalmip2016sampleBased}}. Specifically, the polytope was built using 1000 randomly generated rays, 
yielding 6592 vertices. This fact reinforces the findings of the previous example regarding the efficiency of our approach in comparison to purely geometric set approximation alternatives. 
Moreover, the technique in \cite{yalmip2016sampleBased} unavoidably introduces some degree of conservatism, which is avoided using our method. 
While a spherical $\Omega$ theoretically has an infinite number of vertices, Theorem~\ref{th:convergence} guarantees that, in view of the 
assumption 
$\eta \geq \varepsilon>0$, to find a quadratic control Lyapunov function and associated linear controller is sufficient to examine only a finite number of vertices. 



%% file: main.bbl
\begin{thebibliography}{10}
\providecommand{\url}[1]{#1}
\csname url@samestyle\endcsname
\providecommand{\newblock}{\relax}
\providecommand{\bibinfo}[2]{#2}
\providecommand{\BIBentrySTDinterwordspacing}{\spaceskip=0pt\relax}
\providecommand{\BIBentryALTinterwordstretchfactor}{4}
\providecommand{\BIBentryALTinterwordspacing}{\spaceskip=\fontdimen2\font plus
\BIBentryALTinterwordstretchfactor\fontdimen3\font minus
  \fontdimen4\font\relax}
\providecommand{\BIBforeignlanguage}[2]{{%
\expandafter\ifx\csname l@#1\endcsname\relax
\typeout{** WARNING: IEEEtran.bst: No hyphenation pattern has been}%
\typeout{** loaded for the language `#1'. Using the pattern for}%
\typeout{** the default language instead.}%
\else
\language=\csname l@#1\endcsname
\fi
#2}}
\providecommand{\BIBdecl}{\relax}
\BIBdecl

\bibitem{blanchini2015set}
F.~Blanchini and S.~Miani, \emph{Set-theoretic methods in control}.\hskip 1em
  plus 0.5em minus 0.4em\relax Birkh{\"a}user, 2015.

\bibitem{biswas2005survey}
P.~Biswas, P.~Grieder, J.~L{\"o}fberg, and M.~Morari, ``A survey on stability
  analysis of discrete-time piecewise affine systems,'' \emph{IFAC Proc.
  Volumes}, vol.~38, no.~1, pp. 283--294, 2005.

\bibitem{kothare1996robust}
M.~V. Kothare, V.~Balakrishnan, and M.~Morari, ``Robust constrained model
  predictive control using linear matrix inequalities,'' \emph{Automatica},
  vol.~32, no.~10, pp. 1361--1379, 1996.

\bibitem{parrilo2000structured}
P.~A. Parrilo, \emph{Structured semidefinite programs and semialgebraic
  geometry methods in robustness and optimization}.\hskip 1em plus 0.5em minus
  0.4em\relax Caltech, 2000.

\bibitem{WuPrajna2004}
F.~Wu and S.~Prajna, ``{A new solution approach to polynomial {LPV} system
  analysis and synthesis},'' in \emph{IEEE ACC}, 2004, pp. 1362--1367.

\bibitem{gabyetal2022}
N.~Gaby, F.~Zhang, and X.~Ye, ``Lyapunov-net: A deep neural network
  architecture for {L}yapunov function approximation,'' in \emph{IEEE CDC},
  2022.

\bibitem{munser2021synthesis}
L.~Munser, G.~Devadze, and S.~Streif, ``Synthesis of {L}yapunov functions using
  formal verification,'' \emph{arXiv:2112.01835}, 2021.

\bibitem{ravanbakhsh2019learning}
H.~Ravanbakhsh and S.~Sankaranarayanan, ``Learning control {L}yapunov functions
  from counterexamples and demonstrations,'' \emph{Autonomous Robots}, vol.~43,
  pp. 275--307, 2019.

\bibitem{mcgough2010symbolic}
J.~S. McGough, A.~W. Christianson, and R.~C. Hoover, ``Symbolic computation of
  {L}yapunov functions using evolutionary algorithms,'' in \emph{Proc.~of the
  12th {IASTED} Int. Conf.}, vol.~15, 2010, pp. 508--515.

\bibitem{khansari2017learning}
S.~M. Khansari-Zadeh and O.~Khatib, ``Learning potential functions from human
  demonstrations with encapsulated dynamic and compliant behaviors,''
  \emph{Autonomous Robots}, 2017.

\bibitem{mordatch2014combining}
I.~Mordatch and E.~Todorov, ``Combining the benefits of function approximation
  and trajectory optimization.'' in \emph{Robotics: Science and Systems},
  vol.~4, 2014, p.~23.

\bibitem{solar2006combinatorial}
A.~Solar-Lezama, L.~Tancau, R.~Bodik, S.~Seshia, and V.~Saraswat,
  ``Combinatorial sketching for finite programs,'' in \emph{ACM ASPLOS}, 2006.

\bibitem{abate2017automated}
A.~Abate, I.~Bessa, D.~Cattaruzza, L.~Cordeiro, C.~David, P.~Kesseli,
  D.~Kroening, and E.~Polgreen, ``Automated formal synthesis of digital
  controllers for state-space physical plants,'' in \emph{Proc. of the Int.
  Conf. on Computer Aided Verification}.\hskip 1em plus 0.5em minus 0.4em\relax
  Springer, 2017.

\bibitem{berger2022learning}
G.~O. Berger and S.~Sankaranarayanan, ``Learning fixed-complexity polyhedral
  {L}yapunov functions from counterexamples,'' in \emph{IEEE CDC}, 2022.

\bibitem{chang2019neural}
Y.-C. Chang, N.~Roohi, and S.~Gao, ``Neural {L}yapunov control,'' in
  \emph{NIPS}, 2019.

\bibitem{dai2020counter}
H.~Dai, B.~Landry, M.~Pavone, and R.~Tedrake, ``Counter-example guided
  synthesis of neural network {L}yapunov functions for piecewise linear
  systems,'' in \emph{IEEE CDC}, 2020.

\bibitem{Abate2023Neural}
A.~Abate, A.~Edwards, and M.~Giacobbe, ``Neural abstractions,''
  \emph{arXiv:2301.11683}, 2023.

\bibitem{AhmedPeruffoAbate2020}
D.~Ahmed, A.~Peruffo, and A.~Abate, ``{Automated and Sound Synthesis of
  {Lyapunov} Functions with {SMT} Solvers},'' in \emph{Tools and Algorithms for
  the Construction and Analysis of Systems}.\hskip 1em plus 0.5em minus
  0.4em\relax Springer, 2020, pp. 97--114.

\bibitem{kato2013perturbation}
T.~Kato, \emph{Perturbation theory for linear operators}.\hskip 1em plus 0.5em
  minus 0.4em\relax Springer, 2013.

\bibitem{duan2013LMIs}
G.-R. Duan and H.-H. Yu, \emph{{LMI}s in control systems: {A}nalysis, design
  and applications}.\hskip 1em plus 0.5em minus 0.4em\relax CRC Press, 2013.

\bibitem{hornjohnson1991}
T.~A. Horn and C.~R. Johnson, \emph{Topics in matrix analysis}.\hskip 1em plus
  0.5em minus 0.4em\relax Cambridge University Press, 1991.

\bibitem{alessio2007squaring}
A.~Alessio, M.~Lazar, A.~Bemporad, and W.~Heemels, ``Squaring the circle: An
  algorithm for generating polyhedral invariant sets from ellipsoidal ones,''
  \emph{Automatica}, vol.~43, pp. 2096--2103, 2007.

\bibitem{malherbe2017global}
C.~Malherbe and N.~Vayatis, ``Global optimization of {L}ipschitz functions,''
  in \emph{ACM ICML}, 2017.

\bibitem{Lofberg2004YALMIP}
J.~L{\"{o}}fberg, ``{YALMIP}: A toolbox for modeling and optimization in
  matlab,'' in \emph{In Proc. of the CACSD Conf.}, 2004.

\bibitem{Herceg2013MPT3}
M.~Herceg, M.~Kvasnica, C.~Jones, and M.~Morari, ``{Multi-Parametric Toolbox
  3.0},'' in \emph{IEEE ECC}, 2013.

\bibitem{stripinis2022directgo}
L.~Stripinis and R.~Paulavi{\v{c}}ius, ``{DIRECTGO}: A new {DIRECT}-type
  {MATLAB} toolbox for derivative-free global optimization,'' \emph{ACM Trans.
  on Mathematical Software}, 2022.

\bibitem{yalmip2016sampleBased}
\BIBentryALTinterwordspacing
J.~Löfberg, ``Sample-based robust optimization,'' 2016. [Online]. Available:
  \url{https://yalmip.github.io/example/scenariosampling/}
\BIBentrySTDinterwordspacing

\end{thebibliography}
